\documentclass[11pt]{amsart}
\usepackage{mathrsfs} 
\usepackage[margin=1.25in]{geometry}
\usepackage{ifthen}
\usepackage{graphicx}
\usepackage{epsfig}
\usepackage{dsfont, amssymb}
\usepackage{amsfonts,dsfont,amssymb,amsthm,stmaryrd,bbm}
\usepackage{amsmath}

\usepackage{xcolor}

\usepackage[toc,page]{appendix} 
\usepackage{hyperref,mathtools}
\hypersetup{colorlinks=true,pdftitle=""
}

\let\originallesssim\lesssim
\let\originalgtrsim\gtrsim

\DeclareRobustCommand{\lesssim}{%
  \mathrel{\mathpalette\lowersim\originallesssim}%
}
\DeclareRobustCommand{\gtrsim}{%
  \mathrel{\mathpalette\lowersim\originalgtrsim}%
}

\makeatletter
\newcommand{\lowersim}[2]{%
  \sbox\z@{$#1<$}%
  \raisebox{-\dimexpr\height-\ht\z@}{$\m@th#1#2$}%
}
\makeatother

\newtheorem{conj}{Conjecture}[section]
\newtheorem{thm}{Theorem}[section]

\newtheorem{remark}[thm]{Remark}
\newtheorem{lem}[thm]{Lemma}
\newtheorem{prop}[thm]{Proposition}
\newtheorem{coro}[thm]{Corollary}
\newtheorem{ques}[conj]{Question}

\newtheorem{defn}[thm]{Definition}
\newtheorem{cor}[thm]{Corollary}

\newcommand\independent{\protect\mathpalette{\protect\independent}{\perp}} 
\def\independent#1#2{\mathrel{\rlap{$#1#2$}\mkern2mu{#1#2}}}

\DeclareMathOperator{\Var}{Var}

\def\Var{{\rm Var}}

\def\phi{\varphi}

\def\bee{\begin{eqnarray*}}
\def\ene{\end{eqnarray*}}
\usepackage{mathtools}

\begin{document}

\title[Discrete Complement of Lypunov's Inequality]{A discrete complement of Lyapunov's inequality and its information theoretic consequences}
\author[J. Melbourne \& G. Palafox-Castillo]{James Melbourne and Gerardo Palafox-Castillo}

\begin{abstract}
    We establish a reversal of Lyapunov's inequality for monotone log-concave sequences, settling a conjecture of Havrilla-Tkocz and Melbourne-Tkocz. A strengthened version of the same conjecture is disproved through counter example.  We also derive several information theoretic inequalities as consequences.  In particular sharp bounds are derived for the varentropy, R\'enyi entropies, and the concentration of information of monotone log-concave random variables.  Moreover, the majorization approach utilized in the proof of the main theorem, is applied to derive analogous information theoretic results in the symmetric setting, where the Lyapunov reversal is known to fail.
\end{abstract}

\maketitle



\section{Introduction}

In this paper we prove the following reversal of Lyapunov's inequality\footnote{By Lyapunov's inequality we refer to the fact that $p \mapsto \log \| f\|_p^p$ is convex in $p$ for a general measurable function $f$ and measure.}, conjectured in \cite{melbourne2020reversals} and \cite{havrilla2019sharp}.
\begin{thm} \label{thm: conjecture}
    For $x$, a monotone, log-concave sequence in $\ell_1$, the function
    \begin{align*}
        t \mapsto \log \left(t \sum_i x_i^t \right)
    \end{align*}
    is strictly concave for $t \in (0,\infty)$.
\end{thm}

This is anticipated by affirmative results in the continuous setting dating back to Cohn \cite{Coh69} on $\mathbb{R}$, and Borell \cite{Bor73a} in $\mathbb{R}^d$.  However in contrast to the continuous setting, the requirement that $x$ is monotone cannot be dropped\footnote{Note that for continuous variables on $\mathbb{R}$, proving the result for monotone variables is equivalent to the general result since log-concavity is preserved under rearrangement, see for example \cite{melbourne2019rearrangement}.}, see \cite{melbourne2020reversals} for examples. Moreover, we will also provide a counter example to a strengthening of Theorem \ref{thm: conjecture} conjectured in \cite{melbourne2020reversals, havrilla2019sharp}, further differentiating the continuous and discrete settings.

The main novelty in the proof is to establish a majorization between the distribution function of a monotone log-concave sequence and its geometric counterpart.
Though we will not expound upon this outside of its application to this proof, it can be understood as a second order analog of the distributional majorization lemma utilized in \cite{NP00,melbourne2021transport}.  This alongside some further reductions, leaves one needing only the special case of a geometric sequence, which can be approached with direct computation, to complete the proof of Theorem \ref{thm: conjecture}.  

This effort fits within a more general pursuit, developing discrete analogs for the continuous convexity theory, which in recent investigation has connected information theory and convex geometry (see  \cite{MMX17:1} for background).  One instantiation is the effort to understand the behavior of the entropy of discrete variables under independent summation, see \cite{JY10, madiman2019majorization, melbourne2020reversals, madiman2021bernoulli, bobkov2020concentration}.  Another is the pursuit of discrete Brunn-Minkowski type inequalities \cite{GRST14, OV12, marsiglietti2020geometric, klartag2019poisson, gozlan2019transport, GG01, slomka2020remark, halikias2019discrete}.  In fact, in information theoretic language, the Brunn-Minkowski inequality can be understood as a ``R\'enyi entropy power'' inequality, see \cite{BC15:2, li2020further, BM16, marsiglietti2018entropy, MMX17:2, li2018renyi, rioul2018renyi, RS16}, and in this sense, as an information theoretic inequality as well.

We will see that Theorem \ref{thm: conjecture} yields several information theoretic consequences for discrete monotone variables.   We obtain sharp bounds on the varentropy to be compared to its continuous analog \cite{FMW16}, and utilize this to derive concentration of information, analogous to \cite{BM11:aop, FMW16}.  We give sharp reversals of the monotonicity of R\'enyi entropy general parameters, augmenting the recently obtained comparisons for the $\infty$-R\'enyi entropy given in \cite{melbourne2020reversals}, and as a consequence we obtain a sharp reverse entropy power type inequality for iid variables, tightening a result from \cite{melbourne2020reversals}. We mention that this reverse entropy power is the discrete analog of an entropic Rogers-Shephard inequality pursued by Madiman and Kontoyannis in \cite{madiman2016entropy}.  We also obtain a sharp comparison between the value of a log-concave sequence at its mean, and the value at its mode which we compare with the classical result of Darroch \cite{darroch1964distribution} for Bernoulli sums.  We will also obtain as a Corollary of our arguments that for the monotone log-concave variables of a fixed $p$-R\'enyi entropy, the geometric distribution has maximal $q$-R\'enyi entropy for $q \geq p$ and minimal $q$-R\'enyi entropy for $ q \leq p$.

As mentioned, Theorem \ref{thm: conjecture} can fail without the assumption of monotonicity.  In particular, symmetric variables do not necessarily satisfy the conclusion of Theorem \ref{thm: conjecture}.  However we will demonstrate that the majorization techniques used are robust enough to be applied in the symmetric case, and we use them to deliver sharp R\'enyi entropy comparisons, varentropy bounds, and concentration of information results in the symmetric setting.  We also establish the ``symmetric geometric'' distribution as the maximal (resp. minimal) $q$-R\'enyi entropy distribution for fixed $p$-R\'enyi entropy among discrete symmetric log-concave variables for $q \geq p$ (resp. $q \leq p$).

Let us outline the paper.  In Section \ref{sec: Applications}, we will define notation and derive applications of the main theorem.  In Section \ref{Sec: Proof of Theorem} we give the proof of Theorem \ref{thm: conjecture}.  In Section \ref{Sec: Disprove theorem} we give a counter example to the strengthening of Theorem \ref{thm: conjecture} conjectured in \cite{havrilla2019sharp, melbourne2020reversals}, while in Section \ref{sec: Symmetric variables}, we derive analogs of the consequences in Section \ref{sec: Applications} for symmetric log-concave variables.  In the Appendix \ref{sec: Appendix} we recall some elementary results from the theory of majorization for the convenience of the reader.

\section{Applications} \label{sec: Applications}

\subsection{Definitions}
For a real valued random variable $Y$, we let $\mathbb{E} Y$ denote its expectation, and denote its variance $\Var(Y) \coloneqq \mathbb{E} Y^2 - (\mathbb{E}Y)^2$.
\begin{defn}
Let $(E,\mu)$ be a measure space and $X$ an $E$ valued random variable with density function $f$ such that $\mathbb{P}(X \in A) = \int_A f d\mu$. We define the information content $I_X: E \to \mathbb{R}$, as $I_X(x) = - \log f(x)$. 
\end{defn}

To avoid confusion, in sections where we discuss the information content random variable $I_X(X)$, we will avoid the usual abuse of notation and write $H(f)$ for the entropy of a variable $X$.  Conversely, when there is no risk of confusion, and we are considering a single variable $X$, we will omit the subscript and write $I$ for the information content. We write $H_{(\mu)}(X) = H_{(\mu)}(f) \coloneqq \mathbb{E} I(X)$ in the general case. 
For example, when $E$ is discrete, and $\mu$ is the counting measure, $$\mathbb{E} I(X) = H(X)$$ is just the Shannon entropy of $X$.  Observe that when $\mu$ is a probability measure given by a random variable $Y$, then the expectation of the information content is given by the relative entropy (or Kullbeck-Leibler divergence), $H_{(\mu)}(f) = - D(X||Y)$, and the varentropy measures the deviation of $-I(X)$ from $D(X||Y)$.

In physical applications, it may be more natural to write the density of $X$ in terms of a potential $E$, $f(x) = e^{- E(x)}$, in which case,  $\mathbb{E} I(X)$ reflects the average energy of a system, and $V(X)$ the average fluctuation.

\begin{defn}
For a random variable $X$ taking values on a measure space $(E,\mu)$ with density function $f$, define the varentropy functional,
\begin{align}
    V(X) = \mathbb{E} (\log f(X) - \mathbb{E} \log f(X))^2.
\end{align}
\end{defn}

Unless specified, we will consider $E = \mathbb{Z}$ and $\mu$ the standard counting measure, so that the density function $x_n \coloneqq \mathbb{P}(X = n)$ of a variable $X$, can be expressed as a non-negative sequence.

\begin{defn}
A non-negative sequence $x_i$ is log-concave when 
\[
    x_i^2 \geq x_{i-1} x_{i+1}
\]
and $i \leq j \leq k$ with $x_i x_k >0$ implies $x_j > 0$.
\end{defn}

We consider a sequence to be increasing when $x_i x_{i+1} > 0$ implies $x_{i+1} \geq x_i$, and decreasing when $x_i x_{i+1} > 0$ implies $x_{i+1} \leq x_i$.  A sequence is monotone when it is either increasing or decreasing.  
\begin{defn}
A random variable $X$ taking values in $\mathbb{Z}$ is log-concave when the sequence  $x_i \coloneqq \mathbb{P}(X = i)$ is log-concave.   The variable $X$ is monotone when the sequence $x_i$ is monotone.
\end{defn}

We say a non-negative sequence $x_i$ belongs to $\ell_p$ when $\sum_{i \in \mathbb{Z}} x_i^p < \infty$.  Note that when $x_i$ is log-concave, $x_i$ belonging to $\ell_1$ implies that $x_i$ belongs to $\ell_p$ for all $p \in (0,\infty)$.

\subsection{Varentropy bounds}

\begin{thm} \label{thm: varentropy bounds}
    For $X$ a monotone log-concave variable taking values in $\mathbb{Z}$ with the usual counting measure, 
    \begin{align*}
        V(X) < 1.
    \end{align*}
\end{thm}

\begin{proof}
Define $\Psi(t) = \log \left(t \sum_n f^t(n) \right)$, then
\begin{align*}
    \Psi'(t) =  \frac{ \sum_n \log f(n) f^t(n) }{\sum_n f^t(n)}  + \frac 1 t
\end{align*}
and
\begin{align*}
    \Psi''(t) =  \frac{\left( \sum_n f^t(n) \right)\left( \sum_n \log^2 f(n) f^t(n) \right) - \left( \sum_n \log f(n) f^t(n) \right)^2}{ \left( \sum_n f^t(n) \right)^2} - \frac{1}{t^2}
\end{align*}
By concavity of $\Psi$, $\Psi''(1) = V(X) - 1 < 0$, and our result follows.
\end{proof}

The bound is sharp, the varentropy of a geometric distribution $Z_p$ with parameter $p$, can be explicitly computed as $V(Z_p) = \left( \frac{(1-p) \log (1-p)}{p} \right)^2$ which tends to $1$ with $p \to 0$.

\subsection{R\'enyi entropy comparisons}
\begin{defn}
For $X$ a random variable on $\mathbb{Z}$ , and $p \in (0,1) \cup (1,\infty)$ define
\begin{align*}
    H_p(X) \coloneqq \frac{ \log \left( \sum_i x_i^p \right)}{1-p},
\end{align*}
where $x_i \coloneqq \mathbb{P}(X = i)$.
Let $H_1(X) \coloneqq H(X) = - \sum_i x_i \log x_i$ and $H_\infty(X) = - \log \|x \|_\infty$ where $\|x\|_\infty \coloneqq \max_i x_i$ and $H_0(X) = \# \{i : x_i > 0\}$.
\end{defn}
\begin{thm} \label{thm: Renyi entropy reversal}
When $X$ is a monotone and log-concave variable taking values in $\mathbb{Z}$ then $p > q >0$ implies,
\begin{align*}
    H_p(X) > H_q(X) + \log \left( \frac{p^{\frac{1}{p-1}}}{q^{\frac{1}{q-1}}} \right)
\end{align*}
\end{thm}

\begin{proof}
Let $x_i = \mathbb{P}(X = i)$.
We prove the case $p > q > 1$, the other cases can be treated similarly.  Letting $\lambda = \frac{q-1}{p-1}$, $q = \lambda p + (1-\lambda) 1$, so that by strict concavity,
\begin{align*}
    \log \left( q \sum_i x_i^q \right) > \lambda \log \left( p \sum_i x_i^p \right) + \log \left( 1 \sum_i x_i^1 \right),
\end{align*}
and our result follows for $p,q \notin \{1,\infty\}$ from this inequality.  Owing to log-concavity there in no difficulty obtaining the limiting cases through continuity.
\end{proof}

Note that when $X_\lambda$ has geometric distribution $\mathbb{P}(X = n) = (1-\lambda) \lambda^n$ for parameter $\lambda \in (0,1)$, its R\'enyi entropy can be computed directly,
$$H_p(X_\lambda) =  \log \left( \frac{(1-\lambda)^p}{1-\lambda^p} \right)^{\frac 1 {1-p}}.$$
Hence,
\begin{align*}
    H_p(X_\lambda) - H_q(X_\lambda)
        &=
            \log \left( \frac{(1-\lambda)^p}{1-\lambda^p} \right)^{\frac 1 {1-p}}
                - 
                \log \left( \frac{(1-\lambda)^q}{1-\lambda^q} \right)^{\frac 1 {1-q}}
                    \\
        &=
            \log \left(\frac{1-\lambda}{1-\lambda^p}\right)^{\frac{1}{1-p}} - \log \left(\frac{1-\lambda}{1-\lambda^q}\right)^{\frac{1}{1-q}},
\end{align*}
which tends to $\log \left( \frac{p^{\frac{1}{p-1}}}{q^{\frac{1}{q-1}}} \right)$ with $\lambda \to 1$.  Thus we see that Theorem \ref{thm: Renyi entropy reversal} is sharp.

The following result is actually a consequence of the R\'enyi entropy comparison derived in \cite{melbourne2020reversals}.  It does not need the assumption of monotonicity.  The result should be compared to the classical result of Darroch \cite{darroch1964distribution}, that states that for independent sums of Bernoulli random variables the distance between the mean and mode is no greater than 1, see also \cite{pitman1997probabilistic, tang2019poisson, madiman2021bernoulli} for background and recent developments on such variables.  For the larger class of log-concave variables such a result is impossible.  For example, a geometric distribution has mode at $0$, but can have arbitrarily large expectation.  However, the result below demonstrates that the value of any log-concave distribution 
at its mean approximates up to an absolute constant $e$, the value of the distribution at its mode.

\begin{coro} \label{cor: value at mean looks like mode}
For $X$ with log-concave density function $f$ with support $A \subseteq \mathbb{Z}$,
\begin{align*}
     \max \{ f ( \lfloor \mathbb{E}X \rfloor ), f( \lceil \mathbb{E} X \rceil ) \} \geq e^{-1} \ \|f\|_\infty
\end{align*}
where $\lfloor \cdot \rfloor$ and $\lceil \cdot \rceil$ denote the usual floor and ceiling.
\end{coro}

Note that inequality is sharp in the sense that the constant $e^{-1}$ cannot be improved, as can be seen by choosing a geometric distribution with large, integer valued mean.

\begin{proof}
Define the log-affine interpolation of $f$,
\begin{align*}
    \tilde{f}(x) = \begin{cases}
        f^{1 - (x - \lfloor x \rfloor)}(\lfloor x \rfloor) f^{x - \lfloor x \rfloor}( \lceil x \rceil)\   &\hbox{ for } x \in co(A),
            \\
            0 &\hbox{ otherwise.}
    \end{cases}
\end{align*}
Then $\log \tilde{f}$ is a concave function on $co(A)$ the convex hull of $A$, and by Jensen's inequality and 
\begin{align*}
   H(f) = - \mathbb{E} \log f(X)  =- \mathbb{E} \log \tilde{f}(X)  \geq - \log \tilde{f}(\mathbb{E} X).
\end{align*}
Using that by Theorem \ref{thm: Renyi entropy reversal} and by Theorem 1.3 of \cite{melbourne2020reversals} in the absence of monotonicity, $H(f) \leq H_\infty(f) + 1$, and inserting the inequality into exponentials we have
\begin{align*}
    \exp( - \log \|f\|_\infty + 1 ) &\geq \exp( - \log \tilde{f}( \mathbb{E} X))
        \\
    \frac{e}{\|f\|_\infty} &\geq \frac 1 {\tilde{f}(\mathbb{E} X)}
\end{align*}
which yields
 \begin{align*}
     e \max \{ f(\lfloor \mathbb{E}X \rfloor), f( \lceil \mathbb{E} X \rceil) \} \geq  f^{1 - (\mathbb{E}X - \lfloor \mathbb{E}X \rfloor)}(\lfloor \mathbb{E}X \rfloor) f^{\mathbb{E}X - \lfloor \mathbb{E}X \rfloor}( \lceil \mathbb{E}X \rceil) \geq \| f\|_\infty
 \end{align*}
\end{proof}

\subsection{Concentration of information content}

\begin{thm} \label{thm: concentration of information}
For $X \sim f$ monotone log-concave variable on $\mathbb{Z}$, for $t >0$
\begin{align*}
    \mathbb{P}( I(X)  \geq H(f) + t ) &\leq (1+t) e^{-t} ,
\end{align*}
and when $t \leq 1$,
\begin{align*}
    \mathbb{P}( I(X) \leq H(f) - t ) \leq (1-t) e^t.
\end{align*}
\end{thm}
Note that when $t = 1$, we obtain $\mathbb{P}( I(X) \leq H(X) - 1) = 0$, implying that $- \log \| f\|_\infty = H_\infty(f) > H(f) -1$ recovers the sharp comparison of min-entropy and Shannon entropy above.  The inequality $H_\infty(X) \geq H(X) - 1$ holds without the monotonicity assumption, see \cite{melbourne2020reversals}. \\

The following is a general and elementary technique for deriving concentration of the information content based on uniform bounds on the varentropy of the ``canonical ensemble''.  In \cite{FMW16}, it is assumed that $X$ takes values in $\mathbb{R}^d$, and has a density with respect to the Lebesgue measure.  We include the proof adapted from \cite{FMW16} below, for the convenience of the reader. 

\begin{lem}[Fradelizi-Madiman-Wang \cite{FMW16}] \label{lem: FMW}
    For a random variable $X$ on $E$ with density $f \in L^\alpha(\mu)$ for all $\alpha >0$, and $X_\alpha \sim \frac{f^\alpha}{\int f^\alpha d\mu}$ satisfying $V(X_\alpha) \leq K$, then for $t > 0$
    \begin{align} \label{eq: upper tail}
        \mathbb{P}( I(X) - H_{(\mu)} (f) \geq t) \leq e^{-K r(t/K)}
    \end{align}
    and 
    \begin{align} \label{eq: lower tail}
        \mathbb{P}( I(X) - H_{(\mu)} (f) \leq - t) \leq e^{-K r(- t /K)}
    \end{align}
    where $r(t) = t - \log( 1+t)$ for $t \geq -1$ and is infinite otherwise.
\end{lem}
The proof is a combination of results from \cite{FMW16}, Theorem 3.1 and Corollary 3.4 in particular.

\begin{proof}
Observe that the function $F(\alpha) = \log \int f^\alpha(x) d\mu(x)$ is infinitely differentiable\footnote{Indeed, the $n$-th derivative of $\alpha \mapsto f^\alpha(x)$, $f^\alpha (\log f)^n$ is measurable as the composition of a measurable function $f$, with a continuous function $x^\alpha (\log x)^n$, and that further, and $|f^{\alpha'} (\log f)^n| \leq \mathbbm{1}_{\{f > 1 \}}f^{\alpha + \varepsilon} C(n,\varepsilon) + \mathbbm{1}_{\{f < 1 \}}f^{\alpha - \varepsilon} c(n,\varepsilon) $ for $\alpha' \in (\alpha - \varepsilon/2, \alpha + \varepsilon/2)$, where $C$ and $c$ are uniform bounds on $(\log x)^n/x^{\varepsilon/2}$ for $x \geq 1$ and $x^{\varepsilon/2} |\log x |^n$ for $x \leq 1$ respectively, so that the requisite domination exists for Lebesgue dominated convergence to pass the derivative and integrals.}
\begin{align*}
    K = \sup_{\alpha > 0} V(X_\alpha) = \sup_{\alpha > 0} \alpha^2 {F''(\alpha)}.
\end{align*}
By applying $F''(t) \leq  K/t^2$ to the Taylor expansion,
\begin{align*}
    F(\alpha) = F(1) + (\alpha -1) F'(1) + \int_1^\alpha (\alpha- t) F''(t) dt
\end{align*}
yields
\begin{align} \label{eq: F alpha bound}
    F(\alpha) = F(1) + (\alpha - 1) F'(1) + K(\alpha - 1 - \log \alpha ).
\end{align}
With the substitution $\alpha = 1 - \beta$, and the insertion of $F(1) = 0$, and $F'(1) = - H_{(\mu)}(X)$ we can rewrite \eqref{eq: F alpha bound} as
\begin{align}
    \mathbb{E} \left( e^{\beta (I(X) - H_{(\mu)}(f))} \right) \leq e^{K r(-\beta)}.
\end{align}
For $\beta,t >0$, taking exponentials and applying Markov's inequality,
\begin{align*}
    \mathbb{P}( I(X) - H_{(\mu)}(f) \leq -t) 
        &\leq
            \mathbb{E}\left[ e^{- \beta (I(X) - H_{(\mu)}(f))} \right] e^{- \beta t}
                \\
        &\leq 
            e^{ K(r(\beta) - \frac{\beta t}{K})}
\end{align*}
Standard calculus allows minimization over $\beta$ and yields, $\inf_\beta r(\beta) - \frac{\beta t}{K} = - r(-t/K)$ which gives \eqref{eq: lower tail}.
Applying the same ideas yields \eqref{eq: upper tail} as well.
\end{proof}

\begin{proof}[Proof of Theorem \ref{thm: concentration of information}]
    If $X \sim f$, is log-concave and monotone, then $X_\alpha \sim f_\alpha \coloneqq f^\alpha/\sum_n f^\alpha(n)$ is as well.  Hence by Theorem \ref{thm: varentropy bounds}, $V(X_\alpha) \leq 1$. Applying Lemma \ref{lem: FMW} with $K =1$ yields the result.
\end{proof}

\subsection{Renyi Entropy Power Reversals}
The entropy power inequality, is a fundamental inequality in information theory that gives a sharp lower bound on the amount of entropy increase in summation of continuous independent variables, explicitly taking $\mu$ to be the Lebesgue measure on $\mathbb{R}^d$, and denoting for $X$ with density $f$ with respect to $\mu$ $N(X) = e^{\frac 2 d H_{(\mu)}(f)} $, Shannon's entropy power inequality states that
$N(X+Y) \geq N(X) + N(Y)$ for independent random vectors $X$ and $Y$.  More generally, super-additivity properties of the R\'enyi entropy have been studied, extending the Shannon's EPI, see \cite{BC15:1, RS16, MMX17:1, li2018renyi, li2020further, marsiglietti2018entropy, rioul2018renyi}.   We consider a R\'enyi Entropy Power reversal to be any non-trivial upper bound on the entropy of a sum of random variables, see \cite{BM12:jfa, BM13:goetze, CZ94, Yu08:2, BNT15, XMM16:isit, bobkov2020concentration}.  
\begin{thm} \label{thm: Reverse Renyi EPI}
For $X,Y$ iid, log-concave, and monotone on $\mathbb{Z}$, and $\alpha \in [2,\infty]$
\begin{align*}
    H_\alpha(X -Y) \leq H_\alpha(X) + \log 2.
\end{align*}
\end{thm}

The inequality is a sharp improvement for monotone log-concave variables of Theorem 6.2 of \cite{melbourne2020reversals}, where it is proven that $H_\alpha(X-Y) \leq H_\alpha(X) + \alpha^{\frac{1}{\alpha-1}} \log 2 $ for $X$ and $Y$ iid and log-concave.  To see that the constant $2$ cannot be improved, take $X$ to have density $f(n) = (1-p) p^n$ so that for $n \geq 0$, $f_{X-Y}(n) = \frac{1-p}{1+p} p^{|n|}$. Taking the limit with $p \to 1$ shows the inequality to be sharp.  An alternative motivation for the inequality is its relationship to an entropic generalization conjectured by Madiman and Kontoyannis \cite{madiman2016entropy} of the Rogers-Shephard inequality from convex geometry \cite{RS58:1}, see also \cite{melbourne2020reversals} for further discussion.

 The proof relies on an elementary trick, known to specialists, that $H_2(X) = H_\infty(X-Y)$ holds for iid variables $X$ and $Y$.  We include a proof for completeness and emphasize that this equality is independent of any property of the distribution\footnote{The proof is given for iid log-concave $X$ and $Y$ in \cite{melbourne2020reversals}.}.

 \begin{lem}
     For $X$ and $Y$ iid on $\mathbb{Z}$,
     \begin{align}
         H_2(X) = H_\infty(X-Y).
     \end{align}
 \end{lem}
\begin{proof}
    Let $f$ denote the shared distribution of $X$ and $Y$ and $f_{X-Y}$ the distribution of $X-Y$. We compute directly, 
    \begin{align*}
        \sum_k  f^2 (X = k)
            &=
                \sum_k \mathbb{P}(X = k) \mathbb{P}(Y=k)
                    \\
            &=
                \mathbb{P}(X-Y = 0)
    \end{align*}
After taking logarithms, this shows that
\begin{align}
    H_2( X-Y ) = - \log f_{X-Y}(0),
\end{align}
thus the result follows from demonstrating that $f_{X-Y}(0) = \|f_{X-Y}\|_\infty$.  To this end, we recall the elementary rearrangement inequality (see for instance \cite{marshall2010inequalities}) that for non-negative sequences $x,y \in \ell_2$, 
\begin{align*}
    \sum_i x_i y_i \leq \sum_i x^{\downarrow}_i y^{\downarrow}_i
\end{align*}
where $x^\downarrow$ is the sequence $x$ rearranged in decreasing order.  If we denote $\tau_n f(k) = f(n+k)$ then
\begin{align*}
    f_{X-Y}(n)
        &=
            \sum_k \tau_n f(k) f(k)
                \\
        &\leq 
            \sum_k (\tau_n f)^{\downarrow}(k) f^{\downarrow}(k).
\end{align*}
However since $\tau_n f$ is just a translation of $f$, $(\tau_n f)^\downarrow = f^\downarrow$ and since $\sum_k (f^\downarrow)^2(k) = \sum_k f^2(k) = f_{X-Y}(0)$ our result follows.
\end{proof}
When $\alpha \leq 2$ a constant depending on $\alpha$ can be found using only monotonicity of R\'enyi, see Theorem 6.2 in \cite{melbourne2020reversals}.

\begin{proof}[Proof of Theorem \ref{thm: Reverse Renyi EPI}]
    We use the notation $c(\alpha) = \alpha^{\frac 1 {\alpha - 1}}$, with $c(\infty) \coloneqq 1$.  We prove the case that $\alpha > 2$.  
    \begin{align*}
        H_\alpha(X-Y) 
            &\leq 
                H_\infty(X-Y) + \log \frac{c(\alpha)}{c(\infty)}
                    \\
            &=
                H_2(X) + \log \frac{c(\alpha)}{c(\infty)}
                    \\
            &\leq 
                H_\alpha(X) + \log \frac{c(2)}{c(\alpha)} + \log \frac{c(\alpha)}{c(\infty)}
                    \\
            &=
                H_\alpha(X) + \log 2.
    \end{align*}
\end{proof}

\section{Proof of Theorem \ref{thm: conjecture}} \label{Sec: Proof of Theorem}
For $x$ a monotone, log-concave sequence $\ell_1$ sequence, we denote $\Phi_x: (0,\infty) \to \mathbb{R}$,
\begin{align*}
    \Phi_x(t) \coloneqq \log \left( t \sum_i x_i^t \right).
\end{align*}
To prove that $\Phi_x$ is always strictly concave, we will first start with some reductions.  For $x$ a log-concave sequence and $p > q$ we wish to prove,
\begin{align} \label{eq: concavity statement explicit}
    \Phi_x((1-s) p + s q ) - (1-s) \Phi_x(p) - s \Phi_x(q) \geq 0.
\end{align}
If we denote by $x^q$, the monotone log-concave sequence $(x^q)_i = (x^q_i)$ and $\tilde{p} = p/q$, then by algebraic manipulation the left hand side of \eqref{eq: concavity statement explicit} is exactly
\begin{align} \label{eq: concavity at 1}
\Phi_{x^q} ((1-s) \tilde{p} + s 1 ) - (1-s) \Phi_{x^q} (\tilde{p}) - s \Phi_{x^q}(1) \geq 0.
\end{align}

Additionally observe that for a constant $c > 0$, with $cx$ denoting the sequence $(cx)_i = c x_i$ that $\Phi_{cx}(t) = \Phi_{x}(t) + t \log c$. Thus we can and will without loss of generality assume that $\sum_i x_i = 1$ and need only prove that for $p > 1$, and $s \in (0,1)$
\begin{align} \label{eq: local inequality}
    \Phi_x((1-s) p + s) \geq (1-s) \Phi_x(p) + s.
\end{align}

For the proof of this result we will derive the following lemma.
\begin{lem} \label{lem: geometric minimizes}
    For $x$ a non-Dirac, monotone log-concave probability sequence, $p > 1$, and $q \in (1,p)$, there exists a $\lambda \in (0,1)$ such that the sequence $z$ given by $z_k = (1-\lambda) \lambda^k$ satisfies
    \begin{align} \label{eq: equality for geometric}
        \sum_i x_i^p = \sum_i z_i^p
    \end{align}
    and 
    \begin{align}
        \sum_i x_i^q \geq \sum_i z_i^q
    \end{align}
\end{lem}

As we will see Lemma \ref{lem: geometric minimizes} reduces our problem to proving \eqref{eq: local inequality} for the geometric distribution.   To prove the Lemma, we establish a majorization between the distribution function of a monotone log-concave variable and its geometric counterpart.  

\begin{prop}\label{prop:doublecrossing}
For a sequence $x$, define $F_x (t) \coloneqq \#\lbrace i : x_i > t\rbrace$. Let $x$ be a log-concave, non-increasing sequence, and $z_k = C p^k$ for $C > 0$ and $p \in (0,1)$. Then there exist a finite interval $I$ such that $F_z (t) \leq F_x (t)$ if $t \in I$ and $F_z (t) \geq F_x (t)$ if $t \notin I$.
\end{prop}
\begin{proof}
Define $a \coloneqq \min \{ k : x_k \geq z_k \}$, and $b \coloneqq \sup_k \{ k : x_k \geq z_k \}$. It follows from the log-concavity of $x$ and the log-affinity of $z$ that $\{k : x_k \geq z_k\}$ is a discrete interval.  Thus, the interval\footnote{With the interpretation that $\llbracket a , b \rrbracket = [a, \infty) \cap \mathbb{Z}$ when $b = \infty$.} $$\llbracket a, b \rrbracket = \{k : x_k \geq z_k \}.$$  Let $I = [z_b, x_a)$, with $z_b = 0$ in the case $b = +\infty$.
Let $t \in I$. Two cases will be considered: $t < z_a$ and $z_a \leq t$. First assume $z_b \leq t < z_a$. Let $m = \min\lbrace i : z_i \leq t\rbrace  = F_z (t)$. See that $a < m \leq b$: since $z_b \leq t$, then $m \leq b$ because $m$ is the minimum index such that $z$ satisfies such inequality. Also, if $m \leq a$, then we have $z_m \geq z_a$ because $z$ is decreasing, which gives us both $z_m \leq t$ by definition of $m$ and $z_m > t$ because $z_a > t$. This is a contradiction, thus $a < m$. Finally, since $x_i$ is non-increasing and $a < m \leq b$, we must have 
		\begin{equation}\label{eq:z_y_x}
			z_m \leq t < z_{m-1} \leq x_{m-1} \leq x_{m-2} \leq \dots \leq x_0.
		\end{equation}
From \eqref{eq:z_y_x} we see that $F_x (t) \geq m = F_z(t)$. 
Now, suppose $z_a \leq t < x_a$. Since $z_a \leq t$ then $F_z (t) \leq a$. Now, since $t < x_a$ and $x_i$ is non-increasing, so $x_0 \geq x_1 \geq \dots \geq x_a > t$ and thus $F_x (t) \geq a+1$. Therefore $F_z (t) \leq a < a+1  \leq F_x (t)$.
\end{proof}

The following is a standard fact that holds for general measure spaces.  It follows from the layer-cake representation of a non-negative function, a change of variables, and an application of Fubini-Tonelli.
\begin{prop} \label{prop: Cake layer FUbTON}
Let $X$ be a random variable on the non-negative integers and the sequence $x_i \coloneqq \mathbb{P}(X = i)$, then for $t \geq 1$, $F_x (\lambda)$ as defined in Proposition \ref{prop:doublecrossing} satisfies
\begin{align}
   \sum_{i} x_i^t = t \int_0^\infty \lambda^{t-1} F_x(\lambda) d\lambda.
\end{align}
In particular, $F_x$ is a probability distribution function on $(0,\infty)$ when $x$ is a log-concave probability sequence.
\end{prop}

\begin{lem}\label{convex ordering lemma}
	If $U, V$ are non-negative random variables with densities $f, g$ respectively,  such that $ \mathbb{E}(U) = \mathbb{E}(V) $, and $ f \leq g $  on an interval $I$, and $ f \geq g $ outside $I$, then
	\begin{equation*}
	    \mathbb{E}\left(w(U)\right) \geq \mathbb{E}\left(w(V)\right)
	\end{equation*}
	for any convex function $w$. The inequality reverses if $w$ is concave.
\end{lem}

The proof of Lemma \ref{convex ordering lemma} is classical, and given as Theorem \ref{thm: convex ordering proof for two crossing} in the Appendix for completeness.

\begin{thm}\label{convex ordering consequence}
    If $U, V$ are non-negative random variables with densities $f$ and $g$ respectively, that satisfy $\mathbb{E}(U^p) = \mathbb{E}(V^p)$ for $p > 0$ and $f \leq g$ on an interval $I$, and $f \geq g$ outside of $I$, then 
    \begin{equation*}
        \mathbb{E}\left(w(U^p)\right) \geq \mathbb{E} \left(w(V^p)\right)
    \end{equation*}
    for any convex function $w$. The inequality reverses if $w$ is concave.
\end{thm}

\begin{proof}
Follows directly from Lemma \ref{convex ordering lemma}. Indeed, $U^p$ has density $\tilde{f}(x) = f(x^{\frac 1 p }) x^{\frac{1-p}{p}}p^{-1}$ while $V^p$ has density $\tilde{g}(x) = g(x^{\frac 1 p}) x^{\frac{1-p}{p}} p^{-1}$ so that $U^p$ and $V^p$ satisfy the hypothesis of Lemma \ref{convex ordering lemma} for the interval $I^p \coloneqq \{ w : w = x^p, x \in I \}$.  
\end{proof}

\begin{proof}[Proof of Lemma \ref{lem: geometric minimizes}]

	For $p > 1$, and $x$ not a point mass, $0 < \sum x_{i}^p < \sum x_i = 1$. Then, observe that $\Psi(\lambda) \coloneqq \sum_{k=0}^\infty \left( (1-\lambda) \lambda^k \right)^p = \frac{(1-\lambda)^p}{1-\lambda^p}$. By the intermediate value theorem, since $\Psi(0) = 1$ and $\lim_{\lambda \to 1} \Psi(\lambda) \longrightarrow 0$ as $\lambda \rightarrow 1$ (L'Hospital), there exists $\lambda$ such that \eqref{eq: equality for geometric} holds. \\
	
Let $x$ be log-concave, non-increasing with $\sum_i x_i = 1$, let $z$ be geometric and let $p$ be such that $\sum x_{i}^p = \sum z_{i}^p$. Let $U$ be a random variable with density $F_z$, and $V$ be a random variable with density $F_x$. Since $\sum x_{i}^p = \sum z_{i}^p$ then $\frac{1}{p} \sum x_{i}^p = \frac{1}{p} \sum z_{i}^p$, which implies $\mathbb{E}\left(V^{p-1}\right) = \mathbb{E}\left(U^{p-1}\right)$ by Proposition \ref{prop: Cake layer FUbTON}. With $p > 1$ and $q \in (1,p)$, we have that $g(x) = x^{\frac{q-1}{p-1}}$ is concave, thus $\mathbb{E}\left(g(V^{p-1})\right) \geq \mathbb{E}\left(g(U^{p-1})\right)$ by Proposition \ref{prop:doublecrossing} and Theorem \ref{convex ordering consequence}. Thus $\mathbb{E}\left(V^{q-1}\right) \geq \mathbb{E}\left(U^{q-1}\right)$ and, multiplying both sides by $q$ and using Proposition \ref{prop: Cake layer FUbTON}, we get $\sum x_{i}^q \geq \sum z_{i}^q$.
\end{proof}


The last ingredient of the proof of Theorem \ref{thm: conjecture} is to prove it in the special case that the sequence is geometric.
\begin{prop} \label{prop: Special case geo distr}
Let $z = (z_k)$ be a geometric distribution, i.e., $z_k = (1-\lambda)\lambda^k$ for $\lambda \in (0,1)$ and $k \in \lbrace 0, 1, \dots\rbrace$. Then 
\begin{align}
    \Phi_z (t) = \log \left[ t \sum_i z_{i}^t\right]
\end{align}
is a concave function in $(0, +\infty)$.
\end{prop}
\begin{proof}
See that 
\begin{align*}
    \Phi_z (t) &= \log\left[t (1-\lambda)^t\right] + \log\left[\sum_i (\lambda^t)^i\right]\\
	&= \log t + t \log (1-\lambda) - \log(1-\lambda^t),
\end{align*}
thus 
\begin{align*}
    \Phi_{z}^{''}(t) &= \frac{-1}{t^2} + \frac{\lambda^t}{(1-\lambda^t)^2} \log^2 \lambda \\
    &= \frac{\lambda^t \log^2 \lambda^t - (1-\lambda^t)^2}{\left((1-\lambda^t) t \right)^2}, 
\end{align*}
so $f''(t) \leq 0$ if and only if $\lambda^t \log^2 \lambda^t - (1-\lambda^t)^2 \leq 0$. To prove this, let us consider a variable $y = \lambda^t$ and $g(y) \coloneqq y \log^2 y - (1-y)^2$. To see $g(y) \leq 0$ we proceed in the following way. Clearly, $g(1) = 0$; we want to show this is the maximum value of $g$. This will occur if and only if $g'(1) = 0, g'(y) > 0$ for $y < 1$ and $g'(y) < 0$ for $y > 1$, where
\[g'(y) = 2 \log y + \log^2 y + 2(1-y) .\]
It is clear that $g'(1)=0$ and that there exist some $y < 1$ for which $g'(y) > 0$ (e.g., $y = 1/e$) and some $y > 1$ for which $g'(y) < 0$ (e.g. $y = e$), so it suffices to prove $g'$ is monotone to conclude $y=1$ is the only critical value of $g$. To that end, see that 
\[g''(y) =  \frac{2}{y} + \frac{2 \log y}{y} - 2 \]
is always non-negative. Indeed, $g''(y) \leq 0 \iff \frac{1}{y} + \frac{\log y}{y} \leq 1$, which is equivalent to say $h(y) = \frac{1}{y} + \frac{\log y}{y}$ has a maximum value of 1. This is easy to see as $h'(y) = -\frac{\log y}{y^2}$ is zero at $y=1$, is positive on $(0,1)$ and negative on $(1, +\infty)$, and $h(1) = 1$. 
\end{proof}

\begin{proof}[Proof of Theorem \ref{thm: conjecture}]
By the aforementioned reductions, let $x_i$ be a non-increasing log-concave probability sequence, and $s \in (0,1)$. Then there exists, by Lemma \ref{lem: geometric minimizes}, a geometric distribution $z_i$ such that $\sum z_{i}^p = \sum x_{i}^p$ and moreover for all $q \in (1,p)$, $$\sum x_{i}^q \geq \sum z_{i}^q.$$ Taking $q = s + (1-s)p$, we have by Lemma \ref{lem: geometric minimizes}
$$\Phi_x (s +(1-s)p) \geq \Phi_z (s+(1-s)p).$$  
By Proposition \ref{prop: Special case geo distr} $\Phi_z$ is concave, and hence 
$$\Phi_z ((1-s) p + s) \geq  (1-s) \Phi_z (p) + s.$$ 
Then by hypothesis, $\Phi_x (p) = \Phi_z (p)$, and $\Phi_x (1) = \Phi_z (1) = 0$. Compiling these results gives the following sequence of equalities and inequalities,
\begin{align*}
    \Phi_x (s + (1-s) p) 
        &\geq \Phi_z (s+(1-s)p)
            \\
        &\geq (1-s)  \Phi_z (p)  +s
            \\
        &=  (1-s) \Phi_x (p) +s.
\end{align*}
Hence \eqref{eq: local inequality} holds, and we have concavity for any $\Phi_x$. 
\end{proof}

\begin{cor}
    For $X$ a monotone log-concave random variable, and $Z$ a geometric random variable such that 
    \[ 
        H_p(X) = H_p(Z)
    \]
    then for $q \geq p$,
    \[
        H_q(X) \geq H_q(Z)
    \]
    while
    \[
        H_q(X) \leq H_q(Z)
    \]
    for $ q \leq p$.
\end{cor}

The proof is omitted as it is the same as the symmetric case which is given in detail in Section \ref{sec: Symmetric variables}.

\section{Extensions} \label{Sec: Disprove theorem}

A natural generalization of Theorem 1 was first conjectured in an early version of \cite{havrilla2019sharp}, and reiterated in \cite{melbourne2020reversals}.

\begin{ques}
    For $\gamma >0$ and a positive monotone concave sequence $(y_n)_{n=1}^N$ then the function
    \begin{align*}
        \Phi_y(t) \coloneqq \log \left( (t+\gamma) \sum_{n=1}^N y_n^{\frac{t}{\gamma}} \right)
    \end{align*}
    is concave for $t > - \gamma$.
\end{ques}

However the following counterexample precludes an affirmative answer.

Let $N=2$, $y = \{\lambda,1+\lambda\}$ and consider the points, $\{0,\gamma, 2\gamma\} \subseteq (-\gamma, \infty)$. Concavity of $\Phi_y$ would imply,
\begin{align}
    \exp \Phi_y^2(\gamma) \geq \exp \left( \Phi_y(0) \Phi(2 \gamma) \right),
\end{align}
which is,
\begin{align}
    4 \gamma^2 \left( 2 \lambda + 1 \right) \geq 6 \gamma^2 \left( \lambda^2 + (1+\lambda)^2 \right).
\end{align}
Taking the limit with $\lambda \to 0$ would imply $ 4 \geq 6$.

\section{Symmetric Variables} \label{sec: Symmetric variables}
A random variable on $\mathbb{Z}$ can be symmetric about a point $m \in \mathbb{Z}$ ($f(m +n) = f(m-n)$) or it could be symmetric about $n + \frac 1 2 $ for $n \in \mathbb{Z}$. For example $\mathbb{P}(X = 0) = \mathbb{P}(X=1) = \frac 1 2$ is symmetric about $0 + \frac  1 2$.  In this case, when a log-concave sequence $(x_i)_{i \in \mathbb{Z}}$ is symmetric about a point $n + \frac 1 2$, 
\begin{align*}
    \log \left( t \sum_i x_i^t \right) = \log \left( t \sum_{i > n} x_i^t\right) + \log 2
\end{align*}
is concave by Theorem \ref{thm: conjecture} as $(x_i)_{i > n}$ is monotone and log-concave.  Thus, we have the following corollary.
\begin{coro} \label{cor: Symmetric about a half point}
For $(x_i)_{i \in \mathbb{Z}}$ an $\ell_1$ log-concave sequence, symmetric about a point $n + \frac 1 2$,
\begin{align}\label{eq: concavity restated}
    t \mapsto \log \left( t \sum_i x_i^t \right)
\end{align}
is concave in $t$.  Moreover, if $X$ is a random variable satisfying $\mathbb{P}(X =i) = x_i$ then
\begin{align*}
    V(X) < 1,
\end{align*}
and
\begin{align*}
    H_p(X) > H_q(X) + \log \left( \frac{p^{p-1}}{q^{q-1}} \right).
\end{align*}
If $f$ denotes the density of $X$ then,
\begin{align*}
    \mathbb{P}( I(X)  \geq H(f) + t ) &\leq (1+t) e^{-t} 
\end{align*}
and when $t \leq 1$,
\begin{align*}
    \mathbb{P}( I(X) \leq H(f) - t ) \leq (1-t) e^t.
\end{align*}
\end{coro}
\begin{remark}
See that this implies Theorem \ref{thm: concentration of information} is valid for sequences symmetric about a point $n + \frac{1}{2}$.
\end{remark}

However, in the case that $(x_i)$ is symmetric about a point $n \in \mathbb{Z}$, the concavity of \eqref{eq: concavity restated} is known to fail.  In spite of this, we show in the sequel that arguments from the proof of Theorem \ref{thm: conjecture} are able to recover sharp bounds on  the varentropy and the R\'enyi entropy in this setting.
\begin{defn}
    A sequence $z$ is symmetric geometric when there exists $t \in (0,1)$ and $C >0$ such that
    \begin{align*}
        z_n = C \lambda^{|n|}
    \end{align*}
    for $n \in \mathbb{Z}$.  When $C = \frac{1-\lambda}{1+\lambda}$ the sequence defines a probability distribution.  A random variable $Z$ is symmetric geometric when
    \begin{align*}
        \mathbb{P}(Z = n) = \frac{1-\lambda}{1+\lambda} \lambda^{|n|}.
    \end{align*}
\end{defn}
Given $p \in (0,\infty)$, and $X$ symmetric and log-concave, there exists $Z$ symmetric geometric, such that $\sum_n f_X^p(n) = \sum_n f_Z^p(n)$

\begin{prop}\label{prop symmetric equality}
        Let $x_i$ be a probability distribution over $\mathbb{Z}$. For $p \neq 1$, there exists a geometric sequence $z_i = \frac{1-\lambda}{1+\lambda} \lambda^{|i|}$ such that $\sum x_{i}^p = \sum z_{i}^p$.    
\end{prop}    
\begin{proof}    
        Suppose $p > 1$. Then $0 \leq x_{i}^p \leq x_i$, so $0 \leq \sum_i x_{i}^p \leq \sum_i x_i = 1$. Now consider a geometric symmetric sequence $z_i$ with parameter $q$, and see that
        \begin{align*}
            \sum_i z_{i}^p &= \sum_i \left(\frac{1-\lambda}{1+\lambda} \lambda^{|i|}\right)^p\\
            &= \left(\frac{1-\lambda}{1+\lambda}\right)^p \sum_i \left(\lambda^p\right)^{|i|}\\
            &= \frac{(1-\lambda)^p}{(1+\lambda)^p} \frac{1+\lambda^p}{1-\lambda^p}.
        \end{align*}
        Let $S(\lambda) = \frac{(1-\lambda)^p}{(1+\lambda)^p} \frac{1+\lambda^p}{1-\lambda^p}$. Clearly $S(0) = 1$. while $\lim_{\lambda \to 1} \frac{(1-\lambda)^p}{1-\lambda^p} = 0$. Also $\lim_{\lambda \to 1} \frac{1+\lambda^p}{(1+\lambda)^p} = \frac{1}{2^{p-1}}$. Therefore $\lim_{\lambda \to 1} S(\lambda) = 0$. By the intermediate value theorem, since $S$ is continuous for $\lambda \in (0,1)$, there must be a $\lambda \in (0, 1)$ such that $S(\lambda) = \sum z_{i}^p = \sum_i x_{i}^p$.
        
        A similar approach will handle the case that $p \in (0,1)$.
\end{proof}

\begin{prop}\label{symm double crossing}
        If $x_i$ is non-increasing for $i \geq 0$ and $z$ is symmetric geometric, then there exists a finite interval $I$ such that    
                \begin{align}    
                        F_x (t) \geq F_{z} (t) &\quad t \in I\\    
                        F_x (t) \leq F_{z} (t) &\quad t \notin I    
                \end{align}                                      
\end{prop}                     
\begin{proof}    
        We know the result to be true for $x^* = (x_i)_{i \geq 0}$ and $z^* = (z_i)_{i \geq 0}$ by Proposition \ref{prop:doublecrossing}. Now, see that 
        \[                                                                                                   
                2 F_{x^*} - 1 = F_x,         
        \]          
        and    
        \[    
                2 F_{z^*} - 1 = F_{z}.    
        \]              
        Furthermore, $F_x \geq F_z$ if and only if $2 F_{x^*} - 1 \geq 2 F_{z^*} - 1$ if and only if $F_{x^*} \geq F_{z^*}$. Similarly for $F_x \leq F_z$. Therefore the same interval $I$ given by Proposition \ref{prop:doublecrossing} satisfies our desired inequalities.        
\end{proof}

\begin{lem}\label{lem: symmetric renyi compare}
    Let $X$ be log-concave, symmetric about a point $n \in \mathbb{Z}$. Then there exists a symmetric geometric distribution $Z$, such that $H_p (X) = H_p (Z)$ and
    \begin{align*}
        H_q (X) \geq H_q (Z)
    \end{align*}
    for $q \geq p > 0$, and 
    \begin{align*}
        H_q (X) \leq H_q (Z),
    \end{align*}
    for $0<q \leq p$.
\end{lem}
\begin{proof}
    First, see that $H_p (X) = H_p (Z)$ if and only if $\sum_i x_{i}^p = \sum_i z_{i}^p$ so there must exist such geometric distribution for $p \neq 1$ by Proposition \ref{prop symmetric equality}. Suppose $H_p (X) = H_p (Z)$. Let $U$ be a random variable with density $F_z$, and $V$ be a random variable with density $F_x$. Since $\sum x_{i}^p = \sum z_{i}^p$ then $\frac{1}{p} \sum x_{i}^p = \frac{1}{p} \sum z_{i}^p$, which implies $\mathbb{E}\left(V^{p-1}\right) = \mathbb{E}\left(U^{p-1}\right)$ by Proposition \ref{prop: Cake layer FUbTON}. Let $g(x) = x^{\frac{q-1}{p-1}}$. If $p > 1$ and $q \in (1,p)$, we have that $g(x)$ is concave, thus $\mathbb{E}\left(g(V^{p-1})\right) \geq \mathbb{E}\left(g(U^{p-1})\right)$ by Proposition \ref{prop:doublecrossing} and Theorem \ref{convex ordering consequence}.Thus $\mathbb{E}\left(V^{q-1}\right) \geq \mathbb{E}\left(U^{q-1}\right)$ and, multiplying both sides by $q$ and using Proposition \ref{prop: Cake layer FUbTON}, we get $\sum x_{i}^q \geq \sum z_{i}^q$. The same can be argued if $p < 1$ and $q \in (p,1)$. Now, when $p > 1$ and $q \in (0,1) \cup [p, \infty)$, and when $p < 1$ and $q \in (0, p] \cup (1, \infty)$, $g(x)$ is convex, so by Theorem \ref{convex ordering consequence} the inequality is reversed and $\sum x_{i}^q \leq \sum z_{i}^q$. Now, to pass from the sum to R\'enyi's entropy, we must multiply by $\frac{1}{1-q}$, which reverses the inequality when $q > 1$. So we get
    \begin{align*}
        1 < q < p &\Rightarrow H_q (X) \leq H_q (Z);\\
        q < 1 < p &\Rightarrow H_q (X) \leq H_q (Z); \\
        q < p < 1 &\Rightarrow H_q (X) \leq H_q (Z);\\ 
        1 < p < q &\Rightarrow H_q (X) \geq H_q (Z); \\
        p < 1 < q &\Rightarrow H_q (X) \geq H_q (Z); \\
        p < q < 1 &\Rightarrow H_q (X) \geq H_q (Z).
    \end{align*}
    The limiting cases with $q$ or $p  \in \{1,\infty\}$ can be easily handled using continuity and monotonicity of the R\'enyi entropy as a function of $\alpha \mapsto H_\alpha(X)$ and as a function of the parameter $\lambda$ of a symmetric geometric distribution $Z_\lambda$, $\lambda \to H_\alpha(Z_\lambda)$.
\end{proof}

\begin{thm}
For $X$ log-concave and symmetric about a point $n \in \mathbb{Z}$, $p \geq q$, then
\begin{align*}
     H_q(X) -H_p(X) \leq C(q,p) \coloneqq \sup_Z H_q(Z) - H_p(Z)
\end{align*}
where the supremum is taken over all $Z$ symmetric-geometric.
\end{thm}
\begin{proof}
For any $p \neq 1$ and $q \leq p$ we have, by Lemma \ref{lem: symmetric renyi compare}, a symmetric geometric $Z$ with $H_p (X) = H_p (Z)$ and
\begin{align*}
    H_q (X) \leq H_q (Z),
\end{align*}
which implies 
\begin{align*}
    H_q (X) - H_p (X) \leq H_q (Z) - H_p (Z) \leq \sup_Z H_q (Z) - H_p (Z).
\end{align*}
\end{proof}

\begin{thm} \label{thm: varentropy for symmetric}
For $X$ log-concave and symmetric,
\begin{align*}
    V(X) \leq V_S \coloneqq \sup_Z V(Z)
\end{align*}
where the supremum is taken over all $Z$ symmetric-geometric.
\end{thm}

\begin{proof}
Let $\Psi_X(t) = \log \sum_i x_i^{t + 1}$ where $x_i = \mathbb{P}(X = i)$.  Observe that
\begin{align*}
    \Psi_X(0) &= 1
    \\
    \Psi_X'(0) &= - H(X)
    \\
    \Psi_X''(0) &= V(X)
\end{align*}
Choose $Z$ to be a symmetric geometric distribution satisfying $H(Z) = H(X)$.  By Lemma \ref{lem: symmetric renyi compare}, $H_{1+t}(Z) \leq H_{1+t}(X)$ for $t >0$, which corresponds to $\Psi_Z(t) \geq \Psi_X(t)$ for $t >0$.  By Taylor expansion,
\begin{align*}
    \Psi_X(t) 
        &=
            \Psi_X(0) + \Psi_X'(0) t + \frac{t^2}{2} \Psi_X''(0) + o(t^2)
                \\
        &\leq 
            \Psi_Z(0) + \Psi_Z'(0) t + \frac{t^2}{2} \Psi_Z''(0) + o(t^2)
                \\
        &=
            \Psi_Z(t) 
\end{align*}
Since the Taylor expansions are identical up to linear terms, it follows that $\Psi_Z''(0) = V(Z) \geq V(X) = \Psi_X''(0)$.
\end{proof}

Note that if one expresses the distribution of a symmetric-geometric variable $Z_\lambda$ as $\frac{1-\lambda}{1+\lambda} \lambda^{|k|}$, its varentropy has the closed form expression,
\begin{align*}
    V(Z_\lambda) = \log^2(\lambda) \left(\frac{2\lambda}{1-\lambda} -\left(\frac{2\lambda}{(1-\lambda)(1+\lambda)}\right)^2 \right).
\end{align*}
Numerically, we have $V_S \approx 1.16923$.  This is used for the following corollary.

\begin{cor}
    For $X$ with distribution $f$ log-concave and symmetric on $\mathbb{Z}$, and $t \geq 0$,
    \begin{align*}
        \mathbb{P}( I(X) - H(f) \geq t ) \leq \left( 1 + \frac t V \right)^{V} e^{-t}
    \end{align*}
    and
    \begin{align*}
        \mathbb{P}( I(X) - H(f) \leq - t) \leq \left( 1 - \frac{t}{V} \right)^{V} e^t
    \end{align*}
    where $V \coloneqq V_S \approx 1.16923$ is defined in Theorem \ref{thm: varentropy for symmetric} 
\end{cor}

\begin{proof}
The result follows from combining Lemma \ref{lem: FMW} and Theorem \ref{thm: varentropy for symmetric}.
\end{proof}
\hspace{5cm}

\appendix
\section{Majorization} \label{sec: Appendix}
The following theorem is a well known characterization of the convex order, see \cite{marshall2010inequalities} for proof and further background.
\begin{thm}
For $X$ and $Y$ are random variables on $[0,\infty)$ such that $\mathbb{E}X = \mathbb{E}Y < \infty$, then 
\begin{align} \label{eq: convex order}
     \mathbb{E}\varphi(Y)\geq \mathbb{E}\varphi(X)
\end{align}
holds for all convex functions $\varphi$, if it holds for all $\varphi$ of the form $\varphi(x) = [x - t]_+$ for $t \in (0,\infty)$.  
\end{thm}

When $X$ and $Y$ satisfy \eqref{eq: convex order} we say that $Y$ majorizes $X$ in the convex order, or just that $Y$ majorizes $X$ for short, and write $Y \succ X$.

\begin{thm} \label{thm: convex ordering proof for two crossing}
For non-negative random variables $X \sim f$ and $Y \sim g$ with densities taking values on $[0,\infty)$ such that $\mathbb{E}X = \mathbb{E}Y < \infty$, if there exists an interval $I  \subseteq [0,\infty)$ such that $g \leq f$ on $I$, and $g \geq f$ on $[0,\infty) - I$, then $Y \succ X$.
\end{thm}

\begin{proof}
For $ t \in [0,\infty)$ define $\Psi(t) =  \mathbb{E}[Y-t]_+ - \mathbb{E}[X - t]_+$.  By assumption $\mathbb{E}X = \mathbb{E}Y$, and hence $\Psi(0) = 0$.   By monotone convergence, $\lim_{t \to \infty} \Psi(t) = 0$.  Computing the derivatives of $\Psi$, one obtains $\Psi'(t) =\mathbb{P}(X > t) -  \mathbb{P}(Y > t) $, and $\Psi''(t) = g(t) - f(t)$.  Observe that $\Psi'(0) = 0$, $\lim_{t \to \infty} \Psi'(t) = 0$, and $0 = \mathbb{E}Y - \mathbb{E}X = \int_0^\infty \Psi'(t) dt$.  Thus, $\Phi'$ must be both positive and negative or it is exactly $0$ and the problem is trivial.  As such $\Phi''$ is positve, negative, and then positive.  It follows that $\Phi'$ is positive and then negative, and hence $\Phi \geq 0$ and our result follows.
\end{proof}




\bibliographystyle{plain}
\bibliography{bibibi}

\vspace{2cm}

\noindent James Melbourne \\
Centro de Investigaci\'ones en Matem\'aticas \\
Guanajuato, GTO, MX \\
E-mail: james.melbourne@cimat.mx\\

\vspace{1cm}
\noindent Gerardo Palafox-Castillo \\
Universidad Aut\'onoma de Nuevo Le\'on\\
San Nicol\'as de los Garza, NL, MX\\
E-mail: gerardo.palafoxcstl@uanl.edu.mx
\end{document}